\documentclass[11pt,letterpaper]{article}

\usepackage[margin=1in]{geometry}
\usepackage{amssymb,amsthm,amsmath,amssymb,wrapfig,dsfont,authblk}
\usepackage[dvipsnames]{xcolor}

\usepackage{booktabs} % For formal tables
\usepackage{dsfont}
\usepackage{natbib}
\usepackage{xcolor}
\usepackage{times}
\usepackage[ruled]{algorithm2e} % For algorithms

\usepackage{hyperref}
\hypersetup{
     colorlinks   = true,
     urlcolor    = teal,
	 citecolor = teal,
	 linkcolor = teal
}

\SetAlFnt{\small}
\SetAlCapFnt{\small}
\SetAlCapNameFnt{\small}
\SetAlCapHSkip{0pt}
\IncMargin{-\parindent}

\mathchardef\dash="2D

\renewcommand{\vec}[1]{\boldsymbol{\mathbf{#1}}}
\DeclareMathOperator{\E}{\mathbb{E}}
\renewcommand{\Re}{\mathbb{R}}
\renewcommand{\L}{\mathcal{L}}
\newcommand{\X}{\mathcal{X}}
\newcommand{\Y}{\mathcal{Y}}
\newcommand{\poly}{\mathrm{poly}}
\newcommand{\score}{\text{$s$-score}}
\newcommand{\cscore}{\text{C-score}}
\newcommand{\EXP}{\textsc{EXP3}}
\newcommand{\Hedge}{\textsc{Hedge}}

\renewcommand{\u}{\ell}
\newcommand{\U}{L}

\newtheorem{theorem}{Theorem}[section]
\newtheorem{lemma}[theorem]{Lemma}

\newtheorem{proposition}{Proposition}[section]
\newtheorem{definition}[theorem]{Definition}

\newcommand\numberthis{\addtocounter{equation}{1}\tag{\theequation}}
\DeclareMathOperator{\wins}{>}

\author[]{Nika Haghtalab}
\author[]{Ritesh Noothigattu}
\author[]{Ariel D. Procaccia}
\affil[]{Carnegie Mellon University\\\small\texttt{ \{nhaghtal,rnoothig,arielpro\}@cs.cmu.edu}
	}

\date{}

\begin{document}

\title{Weighted Voting Via No-Regret Learning}

\maketitle

\allowdisplaybreaks

\begin{abstract}
Voting systems typically treat all voters equally. We argue that perhaps they should not: Voters who have supported good choices in the past should be given higher weight than voters who have supported bad ones. To develop a formal framework for desirable weighting schemes, we draw on \emph{no-regret learning}. Specifically, given a voting rule, we wish to design a weighting scheme such that applying the voting rule, with voters weighted by the scheme, leads to choices that are almost as good as those endorsed by the best voter in hindsight. We derive possibility and impossibility results for the existence of such weighting schemes, depending on whether the voting rule and the weighting scheme are deterministic or randomized, as well as on the social choice axioms satisfied by the voting rule. 
\end{abstract}

\section{Introduction}
\label{sec:intro}

In most elections, voters are entitled to equal voting power. This principle underlies the \emph{one person, one vote} doctrine, and is enshrined in the United States Supreme Court ruling in the \emph{Reynolds v. Sims} (1964) case. 

But there are numerous voting systems in which voters do, in fact, have different \emph{weights}. Standard examples include the European Council, where (for certain decisions) the weight of each member country is proportional to its population; and corporate voting procedures where stockholders have one vote per share. Some historical voting systems are even more pertinent: Sweden's 1866 system weighted voters by wealth, giving especially wealthy voters as many as 5000 votes; and a Belgian system, used for a decade at the end of the 19th Century, gave (at least) one vote to each man, (at least) two votes to each educated man, and three votes to men who were both educated and wealthy~\cite{Cong11}. 

The last two examples can be seen as (silly, from a modern viewpoint) attempts to weight voters by \emph{merit}, using wealth and education as measurable proxies thereof. We believe that the basic idea of weighting voters by merit does itself have merit. But we propose to measure a voter's merit by the \emph{quality of his past votes}. That is, a voter who has supported good choices in the past should be given higher weight than a voter who has supported bad ones.  

This high-level scheme is, arguably, most applicable to \emph{repeated aggregation of objective opinions}. For example, consider a group of engineers trying to decide which prototype to develop, based on an objective measure of success such as projected market share. If an engineer supported a certain prototype and it turned out to be a success, she should be given higher weight compared to her peers in future decisions; if it is a failure, her weight should lower. Similar examples include a group of investors selecting companies to invest in; and a group of decision makers in a movie studio choosing movie scripts to produce. Importantly, the recently launched, not-for-profit website \href{http://robovote.org}{RoboVote.org} already provides public access to voting tools for precisely these situations, albeit using methods that always treat all voters equally~\cite{PSZ16}. 

Our goal in this paper, therefore, is to augment existing voting methods with weights, in a way that keeps track of voters' past performance, and guarantees good choices over time. The main conceptual problem we face is the development of a formal framework in which one can reason about desirable weighting schemes; in three words, our solution is \emph{no-regret learning}.

\subsection{Our Approach}

The most basic no-regret learning model involves a set of $n$ \emph{experts}. In each round $t=1,\ldots,T$, the algorithm chooses an expert at random, with probability proportional to their current weights. Then the loss of each expert $i$ at round $t$ is revealed, and the algorithm incurs the expected loss corresponding to its randomized choice. The overall loss (across $T$ rounds) of the algorithm, and of each expert, is defined by summing up the per-round losses. The algorithm's goal is to incur an overall loss that is comparable to the best expert \emph{in hindsight}. Specifically, under a \emph{no-regret learning algorithm}, the \emph{average} (per-round) difference between the algorithm's loss and the loss of the best expert goes to $0$ as $T$ goes to infinity. 

We depart from the classic setting in several ways --- some superficial, and some fundamental. Instead of experts, we have a set of $n$ \emph{voters}. In each round, each voter reveals a \emph{ranking} over a set of alternatives,\footnote{The alternatives can change across rounds, and even their number may vary.} and the loss of each alternative is revealed. In addition, we are given a (possibly randomized) \emph{voting rule}, which receives weighted rankings as input, and outputs the winning alternative. The voting rule is not part of our design space; it is exogenous and fixed throughout the process. The loss of a voter in round $t$ is given by assigning his ranking all the weight (equivalently, imagining that all voters have that ranking), applying the voting rule, and measuring the loss of the winning alternative (or the expected loss, if the rule is randomized). As in the classic setting, our benchmark is the best \emph{voter} in hindsight. 

At first glance, it may seem that our setting easily reduces to the classic one, by treating voters as experts. But our loss is computed by applying the given voting rule to the entire profile of weighted rankings, and therein lies the rub. To develop some intuition, consider the case of two alternatives $a$ and $b$, and the weighted majority rule, which selects $a$ if the total weight of voters who rank $a$ above $b$ is greater than $1/2$, and $b$ otherwise. Suppose that at round $t$, the loss of $a$ is $0$, the loss of $b$ is $1$, and the vote profile and weighting scheme are such that voters ranking $a$ above $b$ have a total weight of $1/2+\epsilon$. Consequently, the rule selects $a$, and our loss at round $t$ is exactly $0$. But if we perturbed the weights slightly, $b$ would be selected, and our loss would jump to $1$. By contrast, in the classic setting the algorithm's loss is obviously continuous in the weights assigned to experts.

An obvious question at this point is whether there is a weighting scheme that would allow us to compete with the best voter in hindsight, \emph{under the weighted majority rule}. Our main research question is much more general:
\begin{quote}
\emph{For which voting rules is there a weighting scheme such that the difference between our average per-round loss and that of the best voter goes to zero as the number of rounds goes to infinity?}
\end{quote}
Ironically, the very formulation of this technical question gives a first answer to our original conceptual question: A desirable weighting scheme, with respect to a given voting rule, is one that gives no-regret guarantees.

\subsection{Our Results}

Analogously to the learning literature, we consider two settings that differ in the type of feedback we receive in each time step, which we can use to adjust the voters' weights. In the \emph{full information} setting, we are informed of the loss of \emph{each} alternative. This would be the case, for example, if the alternatives are companies to invest in. By contrast, in the \emph{partial information} setting, we are only privy to the loss of the selected alternative. This type of feedback is appropriate when the alternatives are product prototypes: we cannot know how successful an undeveloped prototype would have been, but obviously we can measure the success of a prototype that was selected for development. 

In Section~\ref{sec:non-det-wts}, we devise no-regret weighting schemes for both settings, and for any voting rule. Specifically, in the full information setting, we show that for any voting rule there is a weighting scheme with regret $O(\sqrt{T \ln(n)})$; in the partial information setting, the regret guarantee is $O(\sqrt{T n\ln(n)})$. While these results make no assumptions on the voting rule, they also impose no restrictions on the weighting scheme. In particular, the foregoing weighting schemes heavily rely on randomization, that is, they are allowed to sample a weight vector from a distribution in each time step. 

However, deterministic weighting schemes seem more desirable, as they are easier to interpret and explain: a voter's weight depends only on past performance, and not on random decisions made by the scheme. In Section~\ref{sec:det-wts}, therefore, we restrict our attention to deterministic weighting schemes. We find that if the voting rule is itself deterministic, it admits a no-regret weighting scheme if and only if it is \emph{constant on unanimous profiles}. Because this property is not satisfied by any reasonable rule, the theorem should be interpreted as a strong impossibility result. We next consider randomized voting rules, and find that they give rise to much more subtle results, which depend on the properties of the voting rule in question. Specifically, we show that if the voting rule is a \emph{distribution over unilaterals} --- a property satisfied by \emph{randomized positional scoring rules} --- then it admits a deterministic no-regret weighting scheme. By contrast, if the voting rule satisfies a probabilistic version of the famous \emph{Condorcet consistency} axiom, then no-regret guarantees are impossible to achieve through a deterministic weighting scheme.

\subsection{Related Work}

\citet{BM07} provide an excellent overview of basic models and results in no-regret learning; throughout the paper we rely on some important technical results in this space~\cite{freund1995desicion,ACFS02}. Conceptually, our work is superficially related to papers on online ranking, where the algorithm chooses a ranking of objects at each stage. These papers differ from each other in how the loss function is defined, and the type of feedback used. For example, in the model of \citet{RKJ08}, the loss is $0$ if among the top $k$ objects in the ranking there is at least one that is ``relevant'', and $1$ otherwise. \citet{CT15} assume there is a relevance score for each object, and the loss of a ranking is calculated through one of several common measures; the twist is that the algorithm only observes the relevance of the top-ranked object, which is insufficient to even compute the loss of the ranking that it chose (i.e., it is incomparable to bandit feedback). Our setting is quite different, of course: While voters have rankings, our loss is determined by aggregating these rankings via a voting rule. And instead of outputting a ranking over alternatives, our algorithm can only output weights over voters.

We also draw connections to the \emph{computational social choice}~\cite{BCEL+16} literature throughout the paper~\cite{Gib77,CS06,Pro10,Moul83}. For now let us just point to a few papers that share some of the features of our problem. Specifically, there is a significant body of work on weighted voting, in the context of manipulation, control, and bribery in elections~\cite{CSL07,ZPR09,FHH09,FHH15}. And there are papers that study repeated (or \emph{dynamic}) voting~\cite{BP12,PP13}, albeit in settings where the preferences of voters evolve over time.

\section{Preliminaries}

Our work draws on social choice theory and online learning. In this section we present important concepts and results from each of these areas in turn. 

\subsection{Social Choice}
\label{sec:prem-sc}

We consider a set $[n]\triangleq\{1,\ldots,n\}$ of \emph{voters} and a set $A$ of $m$ \emph{alternatives}.
A \emph{vote} $\sigma:A\rightarrow [m]$ is a linear ordering  --- a ranking or permutation --- of the alternatives.
That is, for any vote $\sigma$ and alternative $a$,  $\sigma(a)$ denotes the position of alternative $a$ in vote $\sigma$.
For any $a, b\in A$, $\sigma(a) < \sigma(b)$ indicates that alternative $a$ is preferred to $b$ under vote $\sigma$. We also denote this preference by $a \succ_\sigma b$.
We denote the set of all $m!$ possible votes over $A$ by $\L(A)$.

A \emph{vote profile} $\vec \sigma \in \L(A)^n$ denotes the votes of $n$ voters. 
Furthermore, 
given a vote profile $\vec \sigma \in \L(A)^n$ and a weight vector $\vec w\in \Re_{\geq 0}^n$, we define the \emph{anonymous vote profile corresponding to $\vec \sigma$ and $\vec w$},  denoted $\vec \pi\in [0,1]^{|\L(A)|}$, by setting
\[  \pi_{\sigma} \triangleq \frac{1}{\|\vec w\|_1} \sum_{i=1}^n w_i \mathds{1}_{(\sigma_i = \sigma)}, \quad \forall \sigma \in \L(A).
\]
That is, $\vec \pi$ is an $|\L(A)|$-dimensional vector such that for each vote $\sigma \in \L(A)$, $\pi_\sigma$ is the fraction of the total weight on $\sigma$. 
When  needed, we use $\vec \pi_{\vec \sigma, \vec w}$ to clarify the vote profile and weight vector to which the anonymous vote profile corresponds to. Note that $\vec \pi_{\vec \sigma, \vec w}$ only contains the anonymized information about $\vec \sigma$ and $\vec w$, i.e., the anonymous  vote profile remains the same even when the identities of the voters change.

To aggregate the (weighted) votes into a distribution over alternatives, we next introduce the concept of (anonymous) voting rules. Let $\Delta(\L(A))$ be the set of all possible anonymous vote profiles. Similarly, let $\Delta(A)$ denote the set of all possible distributions over $A$.
An anonymous \emph{voting rule} is a function $f: \Delta(\L(A)) \rightarrow \Delta(A)$ that takes as input an anonymous  vote profile $\vec \pi$ and returns a distribution over the alternatives indicated by a vector $f(\vec \pi)$, where $f(\vec \pi)_a$ is the probability that alternative $a$ is the winner under $\vec \pi$. We say that a voting rule $f$ is \emph{deterministic} if for any $\vec \pi \in \Delta(\L(A))$, $f(\vec \pi) $ has support of size $1$, i.e., there is a unique winner.

One class of anonymous voting rules use the positions of the individual alternatives in order to determine the winners. These rules, collectively called \emph{positional scoring rules}, are defined by a scoring vector $\vec s$ such that $s_1 \geq s_2\geq \dots \geq s_m \geq 0$.
Given a vote $\sigma$, the score of alternative $a \in A$ in $\sigma$ is the score of its position in $\sigma$, i.e., $s_{\sigma(a)}$.
Given an anonymous vote profile $\vec \pi$, 
the score of an alternative is its  overall score in the rankings of $\vec \pi$, that is,
\[ \score_{\vec \pi}(a) \triangleq \sum_{\sigma \in \mathcal{L}(A)} \pi_\sigma s_{\sigma(a)}.
\]
A \emph{deterministic positional scoring rule} chooses the alternative with the highest score, i.e., 
$f(\vec \pi) = \vec e_{a^*}$, where $a^* \in \arg\max_{a \in A} \score_{\vec \pi}(a)$ (tie breaking may be needed). 
On the other hand, a \emph{randomized positional scoring rule} chooses each alternative with probability proportional to its score, i.e., $f(\vec \pi)_a \propto  \score_{\vec \pi}(a)$ for all $a\in A$. Examples of positional scoring rules include \emph{plurality} with $\vec s = (1, 0, \dots, 0)$,   \emph{veto} with $\vec s = (1, \dots,1, 0)$, and  \emph{Borda} with $\vec s = (m-1, m-2,  \dots, 0)$.

Another class of anonymous voting rules use pairwise comparisons between the alternatives to determine the winners.
We are especially interested in the \emph{Copeland} rule, which assigns a score to each alternative based on the number of pairwise majority contests it wins. In an anonymous vote profile $\vec \pi$, we denote by $a  \wins_{\vec \pi} b$ the event that $a$ beats $b$ in a pairwise competition, i.e., $a$ is preferred to $b$ in rankings in $\vec \pi$ that collectively have more than half the weight.  More formally, $\sum_{\sigma \in \L(A)} \pi_\sigma \mathds{1}_{(a \succ_\sigma b)} > 1/2$. We also write $a=_{\vec \pi} b$ if they are tied, i.e.,  $\sum_{\sigma \in \L(A)} \pi_\sigma \mathds{1}_{(a \succ_\sigma b)} = 1/2$.
The Copeland score\footnote{Some refer to this variant of Copeland as $\text{Copeland}_{1/2}$~\cite{FHS08}.} of an alternative is defined by
\[   \cscore_{\vec \pi}(a) \triangleq \left|  \{ b\in A \mid a\wins_{\vec \pi} b \}  \right| + \frac 1 2 \cdot \left|  \{ b\in A \mid a=_{\vec \pi} b \}  \right|.
\]  
The \emph{deterministic Copeland rule} chooses the alternative that has the highest Copeland score (possibly breaking ties), and the \emph{randomized Copeland rule} chooses each alternative with  probability proportional to its Copeland score.

The deterministic Copeland rule satisfies a classic social choice axiom, which we present next. We say that $a\in A$ is a \emph{Condorcet winner} in the vote profile $\vec \pi$ if $a\wins_{\vec \pi} b$ for all $b\in A\setminus\{a\}$. A voting rule is \emph{Condorcet consistent} if it selects a Condorcet winner whenever one exists in the given vote profile. Note that the Copeland score of a Condorcet winner is $m-1$, whereas the Copeland score of any other alternative must be strictly smaller, so a Condorcet winner (if one exists) indeed has maximum Copeland score. 

An anonymous deterministic voting rule $f$ is called \emph{strategyproof}
if for any voter $i\in [n]$, any two vote profiles $\vec \sigma$ and $\vec \sigma'$ for which $\sigma_j = \sigma'_j$ for all $j\neq i$,  
and any weight vector $\vec w$, it holds that either $a = a'$ or $a  \succ_{\sigma_i} a'$, where $a$ and $a'$ are the winning alternatives in $f(\vec \pi_{\vec \sigma, \vec w})$ and $f(\vec \pi_{\vec \sigma', \vec w})$ respectively. In words, whenever a voter reports $\sigma_i'$ instead of $\sigma_i$, the outcome does not improve according to the true ranking $\sigma_i$.
While strategyproofness is a natural property to be desired in a voting rule, the celebrated Gibbard-Satterthwaite Theorem~\cite{Gib73,Sat75} shows that non-dictatorial strategyproof deterministic voting rules do not exist.\footnote{The theorem also requires a range of size at least $3$.}
Subsequently,  \citet{Gib77} extended this result to randomized voting rules.
Before presenting his extension, we introduce some additional definitions. 

Given a \emph{loss function} over the alternatives denoted by a vector $\vec \u \in [0,1]^m$, the expected loss of the alternative
chosen by the rule $f$ under an anonymous vote profile $\vec \pi$ is 
\[  \U_f(\vec \pi, \vec \u) \triangleq \E_{a\sim f(\vec \pi) } [ \u_a]  = f(\vec \pi) \cdot \vec \u.
\]
The higher the loss, the worse the alternative. We say that the loss function $\vec \ell$ is \emph{consistent} with vote $\sigma\in \L(A)$ if for all $a,b\in A$, $a \succ_\sigma b \Leftrightarrow \ell_a < \ell_b$. 
An anonymous randomized rule $f$ is \emph{strategyproof}
if for any voter $i\in [n]$, any two vote profiles $\vec \sigma$ and $\vec \sigma'$ for which $\sigma_j = \sigma'_j$ for all $j\neq i$, any weight vector $\vec w$, 
and any loss function $\vec \u$ that is consistent with $\sigma_i$,  we have $\U_f(\vec \pi_{\vec \sigma, \vec w}, \vec \u) \leq \U_f(\vec \pi_{\vec \sigma', \vec w}, \vec \u)$.

The next proposition is an interpretation of a result of \citet{Gib77} on the structural property shared by all strategyproof randomized voting rules, applied to anonymous voting rules.

\begin{proposition}
Any  strategyproof randomized rule is a distribution over a collection of the following types of rules:
\begin{enumerate}
\item Anonymous Unilaterals: $g$ is an anonymous unilateral if there exists a function $h:\L(A)\rightarrow A$ for which
\[g(\vec \pi) = \sum_{\sigma\in \L(A)} \pi_\sigma \vec e_{h(\sigma)}.
\]
\item Duple:  $g$ is a duple rule if  $|\{a \mid \exists \vec \pi \text{ such that } g(\vec \pi)_a \neq 0 \}| \leq 2$. 

\end{enumerate}
\end{proposition}

Examples of strategyproof randomized voting rules include \emph{randomized positional scoring rules} and the \emph{randomized Copeland} rule, which were previously studied in this context~\cite{CS06,Pro10}. In particular, a randomized positional scoring rule with score vector $\vec s$ is a distribution with probabilities proportional to $s_1, \dots, s_m$ over unilateral rules $g_1, \dots, g_m$, where each $g_i$ corresponds to the function $h_i(\sigma)$ that returns the alternative ranked at position $i$ of $\sigma$.
Similarly,  the randomized Copeland rule is a uniform distribution over duples $g_{a,b}$ for any two different $a, b\in A$, where $g_{a,b}(\vec \pi) = \vec e_a$ if $a\wins_{\vec \pi} b$, $g_{a,b}(\vec \pi) = \vec e_b$ if $b\wins_{\vec \pi} a$, and $(g_{a,b}(\vec \pi))_a = (g_{a,b}(\vec \pi))_b = 1/2$ if $a=_{\vec \pi} b$.

\subsection{Online Learning} \label{sec:prelim_ml}

We next describe the general setting of online learning, also known as  learning from experts.
We consider a game between a \emph{learner} and an \emph{adversary}. There is a set of actions (a.k.a experts) $\X$ available to the learner, a set of actions  $\Y$ available to the adversary, and a loss function $f:\X\times \Y \rightarrow [0, 1]$ that is known to both parties.
In every time step $t\in [T]$, the learner chooses a distribution, denoted by a vector $\vec p^t\in\Delta(\X)$, over the actions in $\X$, and the adversary chooses an action $y^t$ from the set $\Y$. The learner then receives a loss of $f(x^t, y^t)$ for $x^t \sim \vec p^t$.
At this point, the learner receives some feedback regarding the action of the adversary. 
In the \emph{full information} setting, the learner observes $y^t$ before proceeding to the next time step. In the \emph{partial information} setting, the learner only observes the loss $f(x^t, y^t)$.

The \emph{regret} of the algorithm is defined as the difference between its total expected loss and that of the best  fixed action in hindsight. The goal of the learner is to minimize its expected regret, that is, minimize 
\[
\E[Reg_T]\triangleq    \E\left[ \sum_{t=1}^T f(x^t,y^t)  -  \min_{x\in \X}  \sum_{t=1}^T f(x, y^t)   \right],
\]
where the expectation is taken over the choice of $x^t\sim \vec p^t$, and any other random choices made by the algorithm and the adversary. 
An online algorithm is called a \emph{no-regret} algorithm if $\E[Reg_T] \in o(T)$. In words, the average regret of the learner must go to $0$ as $T\rightarrow \infty$.
In general, deterministic algorithms, for which $\| \vec p^t\|_\infty = 1$, can suffer linear regret, because the adversary can choose a sequence of actions $y^1, \dots, y^T$ on which the algorithm makes sub-optimal decisions at every round. Therefore, randomization is one of the  key aspects of no-regret algorithms.

Many online no-regret algorithms are known for the full information and the partial information settings. In particular, the \Hedge\ algorithm~\cite{freund1995desicion} is one of the earliest results in this space for the full information setting. At time $t+1$, \Hedge\ picks each action $x$ with probability $p^{t+1}_x \propto \exp(-\eta F^t(x))$, for 
$F^t(x) = \sum_{s=1}^t f(x, y^s)$ and $\eta = \Theta\left(\sqrt{ 2 \ln(|\X|) \,/\, T} \right)$.
\begin{proposition}[\citet{freund1995desicion}] \label{prop:hedge}
\Hedge\ has regret $\E[Reg_T]\leq O\left( \sqrt{T \ln(|\X|) } \right).$
\end{proposition}
For the partial information setting, the \EXP\ algorithm of \citet{ACFS02} can be thought of as a variant of the \Hedge\ algorithm with importance weighting. In particular, at time $t+1$, \EXP\ picks each action $x$ with probability $p^{t+1}_x \propto \exp\left(  -  \eta \tilde F^t(x) \right)$, for $\eta = \Theta\left(\sqrt{ 2 \ln(|\X|) \, /\, T |\X|} \right)$ and 
\begin{equation} \label{eq:tildeF}
\tilde F^t(x) = \sum_{s=1}^t  \frac{\mathds{1}_{(x^s = x)} f(x, y^s)}{p^s_x}.
\end{equation}
In other words, \EXP\ is is similar to \Hedge, except that instead of taking into account the total loss of an action, $F^t(x)$, it takes into account an  \emph{estimate} of the loss, $\tilde F^t(x)$.

\begin{proposition}[\citet{ACFS02}]\label{prop:exp3}
\EXP\ has regret $\E[Reg_T]\leq O\left( \sqrt{T |\X| \ln(|\X|) } \right).$
\end{proposition}

\section{Problem Formulation}
\label{sec:problem}

In this section, we formulate the question of how one can  \emph{design a weighting scheme that effectively weights the rankings of voters based on the history of their votes and the performance of the selected alternatives}.

We consider a setting where $n$ voters participate in a sequence of elections that are  decided by a known voting rule $f$. In each election, voters submit their rankings over a different set of $m$ alternatives so as to elect a winner. Given an adversarial  sequence of  voters' rankings $\vec \sigma^{1:T}$ and alternative losses $\vec \u^{1:T}$ over a span of $T$ elections, the best voter is the one whose rankings lead to the election of the winners with smallest loss overall. We call this voter \emph{the best voter in hindsight}.
When such a voter is known a priori, the weighting scheme would do well to follow the rankings of this voter throughout the sequence of elections. 
In this case, the overall expected loss of the alternatives chosen under this weighting scheme is
\begin{equation} \label{eq:best-voter}
 \min_{i\in [n]}  \sum_{t=1}^T \U_f(\vec \pi_{\vec \sigma^t, \vec e_i}, \vec \u^t).
\end{equation}

However, when the sequence of elections is not known a priori, the best voter is not known either. In this case, the weighting scheme  has to take an online approach to \emph{weighting the voters' rankings}. 
That is, at each time step $t\leq T$, the weighting scheme chooses a weight vector $\vec w^t$, possibly at random, to weight the rankings of the voters. 
After the election is held, the weighting scheme receives some feedback regarding the quality of the alternatives in that election, typically in the form of the loss  of the elected alternative or that of all alternatives.
Using the feedback, the weighting scheme then re-weights the voters' rankings based on their performance so far.  In this case, the total expected loss  of the weighting scheme is
\[   \sum_{t=1}^T \U_f(\vec \pi_{\vec \sigma^t, \vec w^t}, \vec \u^t).
\]

The type of the feedback  is an important factor in designing a weighting scheme. 
Analogously to the online learning models described in Section~\ref{sec:prelim_ml}, we consider two types of feedback, \emph{full information} and \emph{partial information}. In the full information case, after a winner is selected at time $t$, the quality of all alternatives and rankings of the voters at that round are revealed to the weighting scheme. Note that this information is sufficient for computing the loss of each voter's rankings so far.
On the other hand, in the partial information setting only the loss of the winner is revealed.
More formally, in the full information setting the choice of $\vec w^{t+1}$ can depend on $\vec \sigma^{1:t}$ and $\vec \u^{1:t}$, while in the partial information setting it can only depend on $\vec \sigma^{1:t}$ and $\u^s_{a^s}$ for $s \leq t$, where $a^s$ is the alternative that won the election at time $s$. 

Our goal is to design a weighting scheme that weights the rankings of the voters at each time step, and elects winners with overall expected loss that is almost as small as that of the best voter. We refer to the expected difference between these losses as the expected \emph{regret}. That is,
\[
\E[Reg_T]\triangleq   \E\left[\sum_{t=1}^T \U_f(\vec \pi_{\vec \sigma^t, \vec w^t}, \vec \u^t) - \min_i  \sum_{t=1}^T \U_f(\vec \pi_{\vec \sigma^t, \vec e_i}, \vec \u^t)\right],
\]
where the expectation is taken over any additional source of randomness in the adversarial sequence or the algorithm.
In particular, we seek a weighting scheme for which the average expected regret goes to zero as the time horizon $T$ goes to infinity, at a rate that is polynomial in the number of voters and alternatives.  That is, we wish to achieve $\E[Reg_T] = \poly(n,m)\cdot o(T)$. This is our version of a \emph{no-regret} algorithm.

No doubt the reader has noted that the above problem formulation is closely related to the general setting of online learning. Using the language of online learning introduced in Section~\ref{sec:prelim_ml}, the weight vector $\vec w^t$ corresponds to the learner's action $ x^t$, the vote profile and alternative losses  $(\vec \sigma^t, \vec \u^t)$  correspond to the adversary's action $y^t$, the expected loss of the weighting scheme $\U_f(\vec \pi_{\vec \sigma^t, \vec w^t}, \vec \u^t)$ corresponds to the loss of the learning algorithm $f(x^t, y^t)$, and the best-in-hindsight voter --- or weight vector $\vec e_i$ --- refers to the best-in-hindsight action.

\section{Randomized Weights}	\label{sec:non-det-wts}
In this section, we develop no-regret algorithms for the full information and partial information settings. We essentially require no assumptions on the voting rule, but also impose no restrictions on the weighting scheme. In particular, the weighting scheme may be randomized, that is, the weights can be sampled from a distribution over weight vectors. This allows us to obtain general positive results. 

As we just discussed, our setting is closely related to the classic online learning setting. 
Here, we introduce an algorithm analogous to \Hedge\ that works in the full information setting of Section~\ref{sec:problem} and achieves a total regret of $O(\sqrt{T \ln(n)})$.

\begin{algorithm}
	\SetAlgoNoLine
	\KwIn{Adversarial sequences  $\vec \sigma^{1:T}$ and $\vec \u^{1:T}$, and parameter $\eta = \sqrt{2 \ln n / T}$}
    \For{$t= 1, \dots, T$}{
   Play weight vector $\vec e_i$ with probability
    \[   p^t_i \propto \exp\left(-\eta \sum_{s = 1}^{t-1} \U_f(\vec \pi_{\vec\sigma^s, \vec e_i}, \vec \u^s) \right).
    \]
   Observe $\vec \u^t$ and $\vec \sigma^t$.
    }
\caption{Full information setting, using randomized weights.}
\label{alg:full-rand}
\end{algorithm}

\begin{theorem} \label{thm:full-rand}
For any anonymous voting rule $f$ and $n$ voters, Algorithm~\ref{alg:full-rand} has regret $O(\sqrt{T \ln(n)})$ in the full information setting.
\end{theorem}
\begin{proof}[Proof Sketch.]
At a high level, this algorithm 
only considers weight vectors that correspond to a single voter. At every time step, the algorithm chooses a distribution over such weight vectors and applies the voting rule to one such weight vector that is drawn at random from this distribution. This is equivalent to applying the  \Hedge\ algorithm to a set of actions, each of which is a weight vector that corresponds to a single voter. That is,
\[
\E \left[   \sum_{t=1}^T \U_f(\vec \pi_{\vec \sigma^t, \vec w^t}, \vec \u^t) \right] = \E_{i^t \sim \vec p^t}\left[ \sum_{t=1}^T \U_f(\vec \pi_{\vec \sigma^t, \vec e_{i^t}}, \vec \u^t) \right].
\]
The theorem follows by noting that the loss of the benchmark weighting scheme
(See Equation~\ref{eq:best-voter}) is the smallest loss that one can get from following one such weight vector. 
That is,  by Proposition~\ref{prop:hedge}, the total expected regret is 
\[
\E \left[   \sum_{t=1}^T \U_f(\vec \pi_{\vec \sigma^t, \vec w^t}, \vec \u^t) \right]
- \min_i  \sum_{t=1}^T \U_f(\vec \pi_{\vec \sigma^t, \vec e_i}, \vec \u^t) 
\leq O\left( \sqrt{T \ln(n) }\right) .
\]
\end{proof}

Next, we introduce an algorithm for the partial information setting. One may wonder whether the above approach, i.e., reducing our problem to online learning and using a standard algorithm, directly extends to the partial information setting (with the EXP3 algorithm). The answer is that it does not. In particular, in the classic setting of online learning with partial information feedback, the algorithm observes the action of the adversary and therefore can compute the estimated loss of the action it just played. That is, the algorithm can compute $f(x^t,y^t)$. In our problem setting, however, the weighting scheme only observes $\vec \sigma^t$ and $\ell^t_{a^t}$ for the specific alternative $a^t$ that was elected at this time. Since the losses of other alternatives remain unknown, the weighting scheme cannot even compute the expected loss of the specific voter $i^t$ it selected at time $t$, i.e.,  $\U_f(\vec \pi_{\vec \sigma^t, \vec e_{i^t}}, \vec \u^t)$. Therefore, we cannot directly use the \EXP\ algorithm by imagining that the voters are actions, as we do not obtain the partial information feedback that the algorithm requires.

Nevertheless, the algorithm we introduce here is inspired by \EXP. Fortunately, certain properties that the performance of  \EXP\ relies on still hold in our setting.
In particular, \EXP\ uses $f(x^t,y^t)$ to create an unbiased estimator of the true loss of action $x^t$ over $t$ time steps. As we show, Algorithm~\ref{alg:bandit-rand} also creates an unbiased estimator of the loss of voters in $t$ time steps, using $\u_{a^t}$.

\begin{algorithm}[h!]
	\SetAlgoNoLine
	\KwIn{An adversarial sequences of $\vec \sigma^{1:T}$ and $\vec \u^{1:T}$, and parameter $\eta = \sqrt{2 \ln n / Tn}$.}
Let $\tilde{\vec \U}^0=\mathbf{0}$. \\
\For{$t=1, \dots, T$}{
	\For{$i= 1, \dots, n$}{Let
		\[
		p_i^{t} \propto  \exp( - \eta \tilde \U^{t-1}_i).
		\]      
	}
	Play weight vector $\vec e_{i^t}$ from distribution $\vec p^t$.\\
	Observe the vote profile $\vec \sigma^t$, the alternative $a^t \sim f( \vec \pi_{\vec \sigma^t, \vec e_{i^t}})$, and its loss $\u^t_{a^t}$.\\
	Let $\tilde{\vec \u}^t$ be the vector such that 
	\[   \tilde \u^t_{i^t}  = \frac{  \u^t_{a^t}  }{p^t_{i^t}} \quad \text{ and } \quad \tilde \u^t_i =0  \quad \text{for } i\neq i^t.
	\]
	Let $\tilde{\vec \U}^t = \tilde{ \vec \U}^{t-1} + \tilde{\vec \u}^t$.\\
}
\caption{Partial information setting, using randomized weights.}
\label{alg:bandit-rand}
\end{algorithm}

\begin{theorem}\label{thm:bandit-rand}
For any anonymous voting rule $f$ and $n$ voters, Algorithm~\ref{alg:bandit-rand} has regret $O(\sqrt{T n \ln(n)})$ in the partial information setting.
\end{theorem}

Let us first establish a few crucial properties of Algorithm~\ref{alg:bandit-rand} in preparation for proving  Theorem~\ref{thm:bandit-rand}. 
In the next lemma, we show that $\tilde{\vec \u}^t$ creates an unbiased estimator of the expected loss of the weighting scheme. Similarly, we show that for any voter $i^*$, $\tilde \U^t_{i^*}$ is an unbiased estimator for the loss that the weighting scheme would have received if it followed the rankings of voter $i^*$ throughout the sequence of elections.  
\begin{lemma} \label{lem:estimate}
For any $t$ and any $i^*$ we have
\[ \E\limits_{i^t, a^t} \left[  \sum_{i=1}^n p^t_i \tilde \u^t_i  \right]  = \E_{i^t} \left[  \U_f(\vec \pi_{\vec \sigma^t, \vec e_{i^t}}, \vec \u^t)\right] 
\quad and \quad
\E_{i^t, a^t} \left[ \tilde \U^T_{i^*}  \right]   = \sum_{t=1}^T \U_f(\vec \pi_{\vec \sigma^t, \vec e_{i^*}}, \vec \u^t),
\]
where $i^t\sim \vec p^t$ and $a^t \sim f(\vec \pi_{\vec \sigma^t, \vec e_{i^t}})$.
\end{lemma}
\begin{proof}
For ease of notation, we suppress $t$ when it is clear from the context.
First note that $\tilde{\vec \u}$ is zero in all of its elements, except for $\tilde{\u}_{i^t}$. So, 
\[  \sum_{i=1}^n p_i \tilde \u_i = p_{i^t} \tilde \u_{i^t}  =  p_{i^t} \frac{\u_{a^t}}{p_{i^t}}  = \u_{a^t}.
\]
Therefore, we have
\begin{align*}
\E_{i^t, a^t} \left[  \sum_{i=1}^n p_i \tilde \u_i  \right] = \E_{i^t, a^t}\left[  \u_{a^t}  \right] = \E_{i^t}\left[ \U_f(\vec \pi_{\vec \sigma, \vec e_{i^t}}, \vec \u)  \right].
\end{align*}
For clarity of presentation, let $\tilde{ \vec \u}^{i, a}$ be an alternative representation of $\tilde{\vec \u}$ when $i^t = i$ and $a^t = a$. 
Note that $\u^{i, a}_{i^*} \neq 0$ only if $i^* = i$.
We have 
\begin{align*}
\E_{i^t, a^t} \left[ \tilde \U^T_{i^*}  \right] &= \sum_{t=1}^T \E_{i^t, a^t}\left[  \tilde \u^{i^t, a^t}_{i^*}  \right]
    = \sum_{t=1}^T \sum_{i=1}^n  p^t_i  \E_{a \sim f(\vec \pi_{\vec \sigma^t, \vec e_i})}\left[  \tilde \u^{i, a}_{i^*}  \right]
    = \sum_{t=1}^T p^t_{i^*} \E_{a \sim f(\vec \pi_{\vec \sigma^t, \vec e_{i^*}})}\left[  \frac{\u^t_a}{p^t_{i^*}}  \right]\\
& = \sum_{t=1}^T \E_{a \sim f(\vec \pi_{\vec \sigma^t, \vec e_{i^*}})}\left[ \u^t_a  \right] 
=  \sum_{t=1}^T \U_f(\vec \pi_{\vec \sigma^t, \vec e_{i^*}}, \vec \u^t).
\end{align*}
\end{proof}

\begin{lemma} \label{lem:l2}
For any $t$,  we have
\[  \E_{i^t, a^t}  \left[  \sum_{i=1}^n p^t_i (\tilde \u^t_i)^2 \right] \leq n,
\]
where $i^t\sim \vec p^t$ and $a^t \sim f(\vec \pi_{\vec \sigma^t, \vec e_{i^t}})$.
\end{lemma}
\begin{proof}
For ease of notation, we suppress $t$ when it is clear from the context.
Since $\tilde{\vec \u}$ is zero in all of its elements, except for $\tilde{\u}_{i^t}$, we have
\[  \sum_{i=1}^n p_i (\tilde \u_i)^2 = p_{i^t} (\tilde \u_{i^t})^2  =  p_{i^t} \left( \frac{\u_{a^t}}{p_{i^t}} \right)^2  = \frac{(\u_{a^t})^2}{p_{i^t}}.
\]
Therefore,
\begin{align*}
\E_{i^t, a^t} \left[  \sum_{i=1}^n p_i (\tilde \u_i)^2  \right] = \E_{i^t, a^t}\left[  \frac{(\u_{a^t})^2}{p_{i^t}}  \right] = 
\sum_{i=1}^n p_i \E_{a \sim f(\vec \pi_{\vec \sigma, \vec e_i})} \left[    \frac{(\u_{a})^2}{p_{i}}      \right]  = 
\sum_{i=1}^n \E\limits_{a \sim f(\vec \pi_{\vec \sigma, \vec e_i})}  \left[    (\u_{a})^2      \right]  \leq n.
\end{align*}
\end{proof}

\begin{proof}[Proof of Theorem~\ref{thm:bandit-rand}]
We use a potential function, given by  $\Phi^t \triangleq -\frac 1\eta \ln\left(   \sum_{i=1}^n \exp(- \eta \tilde \U_i^{t-1})  \right).$ We prove the claim by analyzing the expected increase in this potential function at every time step. Note that 
\begin{equation}
\label{eq:phidiff}
\Phi_{t+1} - \Phi_t =  - \frac 1\eta  \ln\left(  \frac{\sum_{i=1}^n \exp( -\eta \tilde \U_i^{t-1} -  \eta \tilde \u^t_i)}{\sum_{i=1}^n \exp( - \eta \tilde \U_i^{t-1})}  \right) =  -\frac 1\eta \ln\left(  \sum_{i=1}^n p^t_i  \exp( - \eta \tilde \u^t_i)  \right).
\end{equation}
Taking the expected increase in the potential function over the random choices of $i^t$ and $a^t$ for all $t=1, \dots, T$, we have
\begin{align*}
\E\left[ \Phi_{T+1} - \Phi_1 \right]  &= 
\sum_{t=1}^T \E_{i^t, a^t} \left[ \Phi_{t+1} - \Phi_t \right] \\
&\geq \sum_{t=1}^T \E_{i^t, a^t} \left[ - \frac 1\eta \ln \left(    \sum_{i=1}^n p^t_i \left(1 - \eta \tilde \u^t_i + \frac 12 \left(\eta \tilde \u^t_i\right)^2 \right)    \right) \right] \\
&= \sum_{t=1}^T \E_{i^t, a^t} \left[ - \frac 1\eta \ln \left(   1 -  \eta \left(  \sum_{i=1}^n p^t_i \tilde \u^t_i -     \frac{\eta}{2} \sum_{i=1}^n p^t_i  \left(\tilde \u^t_i\right)^2 \right)    \right) \right] \\
&\geq \sum_{t=1}^T \E_{i^t, a^t} \left[  \sum_{i=1}^n p^t_i \tilde \u^t_i  -  \frac{\eta}{2} \sum_{i=1}^n p^t_i  \left(\tilde \u^t_i\right)^2 \right] \\
&\geq  \E\left[ \sum_{t=1}^T \U_f(\vec \pi_{\vec \sigma^t, \vec e_{i^t}}, \vec \u^t)\right]   - \frac{\eta T n }{2},\numberthis \label{eq:lower}
\end{align*}
where the second transition follows from Equation~\eqref{eq:phidiff} because for all $x\geq 0$, $e^{-x} \leq 1 - x + \frac{x^2}{2}$, the fourth transition follows from $\ln(1-x)\leq -x$ for all $x\in\mathbb{R}$, and the last transition holds by Lemmas~\ref{lem:estimate} and \ref{lem:l2}.
On the other hand, $\Phi_1 = -\frac 1\eta \ln n$ and for any $i^*$, 
\[
\Phi_{T+1}  \leq -\frac 1\eta \ln\left( \exp(-\eta \tilde{\U}^{T}_{i^*}) \right) = \tilde{\U}^{T}_{i^*}.
\]
Therefore,
\begin{equation}
\E\left[ \Phi_{T+1} - \Phi_1 \right]
\leq  \E\left[ \tilde \U^T_{i^*}  + \frac 1\eta \ln n \right] 
=  \E\left[
\sum_{t=1}^T \U_f(\vec \pi_{\vec \sigma^t, \vec e_{i^*}}, \vec \u^t) + \frac 1\eta \ln n \right]. \label{eq:upper}
\end{equation}
We can now prove the theorem by using Equations~\eqref{eq:lower} and \eqref{eq:upper}, and the parameter value $\eta = \sqrt{2 \ln n / Tn}$:
\[ \E\left[ \sum_{t=1}^T \U_f(\vec \pi_{\vec \sigma^t, \vec e_{i^t}}, \vec \u^t)
-   \min_{i\in[n]}\sum_{t=1}^T \U_f(\vec \pi_{\vec \sigma^t, \vec e_{i}}, \vec \u^t) \right] 
\leq \frac{1}{\eta} \ln n + \frac{\eta T n}{2} \leq \sqrt{2 T n \ln n}.
\]
\end{proof}

\section{Deterministic Weights} \label{sec:det-wts}

One of the key aspects of the weighting schemes we used in the previous section is randomization.
In such weighting schemes,  the weights  of the  voters  not only depend on their performance so far, but also on the algorithm's coin flips. 
In practice, voters would most likely prefer weighting schemes that depend only on their past performance, and are therefore easier to interpret. 

In this section, we focus on designing weighting schemes that are deterministic in nature. Formally, a \emph{deterministic weighting scheme} is an algorithm that at time step $t+1$ deterministically chooses one weight vector $\vec w^{t+1}$ based on the history of play, i.e.,  sequences $\vec \sigma^{1:t}$, $\vec \u^{1:t}$, and $a^{1:t}$.
In this section, we seek an answer to the following question: \emph{``For which voting rules is there a no-regret deterministic weighting scheme?"} 
In contrast to the results established in the previous section, we find that the properties of the voting rule play an important role here. In the remainder of this section, we show possibility and impossibility results for the existence of such weighting schemes under  randomized and deterministic voting rules.

\subsection{Deterministic Voting Rules} \label{sec:det-voting-wt}

We begin our search for deterministic weighting schemes by considering deterministic voting rules.
Note that in this case the winning alternatives are induced deterministically by the weighting scheme, so the weight vector $\vec w^{t+1}$ should be deterministically chosen based on the sequences $\vec \sigma^{1:t}$ and $\vec \u^{1:t}$. 
We establish an impossibility result: Essentially no deterministic weighting scheme is  no-regret for a deterministic voting rule. 
Specifically, we show that a deterministic no-regret weighting scheme exists for a deterministic voting rule if and only if the voting rule is constant on unanimous profiles.

\begin{definition}
A voting rule $f$ is \emph{constant on unanimous profiles} if and only if
$$\forall \sigma, \sigma' \in \L(A) , f(\vec e_\sigma) = f(\vec e_{\sigma'}),$$
where $\vec e_\sigma$ denotes the anonymous vote profile that has all of its weight on ranking $\sigma$.
\end{definition}

\begin{theorem} \label{thm:det-voting-wt}
For any deterministic voting rule $f$, a deterministic weighting scheme with regret $o(T)$ exists if and only if $f$ is constant on unanimous profiles. This is true in both the full information and partial information settings. 
\end{theorem}
\begin{proof}
We first prove that for any voting rule that is constant on unanimous profiles there exists a deterministic weighting scheme that is no-regret.
Consider such a voting rule $f$
and a simple deterministic weighting scheme that
uses weight vector $\vec w^t = \vec e_1$ for every time step $t \leq T$ (so it does not use feedback --- whether full or partial --- at all).
Note that at each time step $t$  and for any voter $i \in [n]$, 
\[ f(\vec \pi_{\vec \sigma^t, \vec w^t}) = f(\vec e_{\sigma^t_1}) = f(\vec e_{\sigma^t_i}) = f(\vec \pi_{\vec \sigma^t, \vec e_i}),
\]
where the second transition holds because $f$ is constant on unanimous profiles. As a result, $L_f(\vec \pi_{\vec \sigma^t, \vec w^t}, \vec \ell^t) =  L_f(\vec \pi_{\vec \sigma^t, \vec e_i}, \vec \ell^t) $.
In words, the total loss of the weighting scheme is the same as the total loss of any individual voter --- this weighting scheme has $0$ regret.

Next, we prove that if $f$ is not constant on unanimous profiles then for any deterministic weighting scheme there is an adversarial sequence of $\vec \sigma^{1:T}$ and $\vec \u^{1:T}$ that leads to regret of $\Omega(T)$, even in the full information setting. 
Take any such voting rule $f$ and let $\tau, \tau' \in \L(A)$ be such that $f(\vec{e}_\tau) \neq f(\vec{e}_{\tau'})$. At time $t$, the adversary chooses $\vec \sigma^t$ and $\vec \u^t$ based on the deterministic weight vector $\vec w^t$ as follows: The adversary sets $\vec \sigma^t$ to be such that $\sigma^t_1 = \tau$  and $\sigma^t_j = \tau'$ for all $j\neq 1$. Let alternative $a^t$ be the winner of profile $\vec \pi_{\vec \sigma^t, \vec w^t}$, i.e., $f(\vec \pi_{\vec \sigma^t, \vec w^t}) = \vec e_{a^t}$. The adversary sets $\u^t_{a^t} = 1$ and $\u^t_x = 0$ for all $x\neq a^t$.
Therefore, the weighting scheme incurs a  loss of $1$ at every step, and its total loss is 
\[ 
\sum_{t=1}^T L_f(\vec \pi_{\vec \sigma^t, \vec w^t}, \vec \ell^t) = \sum_{t=1}^T   \u^t_{a^t}   = T.
\] 

Let us consider the total loss that the ranking of any individual voter incurs.
By design, for any $j >1$,
\[
f(\vec \pi_{\vec \sigma^t, \vec e_1}) = f(\vec{e}_\tau) \neq f(\vec{e}_{\tau'}) = f(\vec \pi_{\vec \sigma^t, \vec e_j}).
\]
Therefore, for at least one  voter $i\in [n]$, $f(\vec \pi_{\vec \sigma^t, \vec e_i})\neq \vec e_{a^t}$. Note that such a voter receives loss of $0$, so  the combined loss of all voters is at most $n-1$. Over all time steps, the total  combined loss of all voters is at most $T(n-1)$. As a result, the best voter incurs a loss of at most  $\frac{(n-1)T}{n}$, i.e.,  the average loss.

We conclude that the regret of the weighting scheme is 
\[Reg_T = \sum_{t=1}^T L_f(\vec \pi_{\vec \sigma^t, \vec w^t}, \vec \ell^t) - \min_{i\in [n]} \sum_{t=1}^T L_f(\vec \pi_{\vec \sigma^t, \vec e_i}, \vec \ell^t) \geq T - \frac{(n-1)T}{n} = \frac Tn.
\]
\end{proof}

\subsection{Randomized Voting Rules}

Theorem~\ref{thm:det-voting-wt} indicates that we need to allow randomness (either in the weighting scheme or in the voting rule) if we wish to have no-regret guarantees. As stated before, we would like to have a deterministic weighting scheme so that the weights of voters are not decided by coin flips. This leaves us with no choice other than having a randomized voting rule. Nonetheless, one might argue in favor of having a deterministic voting rule and a randomized weighting scheme, claiming that it is equivalent because the randomness has simply been shifted from the voting rule to the weights. To that imaginary critic we say that allowing the voting rule to be randomized makes it possible to achieve strategyproofness (see Section~\ref{sec:prem-sc}), which cannot be satisfied by a deterministic voting rule.

The next theorem shows 
that for any voting rule that is a distribution over unilaterals there exist  deterministic weighting schemes that are no-regret. Recall that any randomized positional scoring rule can be represented as a distribution over unilaterals, hence the theorem allows us to design a no-regret weighting scheme for any randomized positional scoring rule. 

The weighting schemes that we use build on Algorithms~\ref{alg:full-rand} and \ref{alg:bandit-rand} directly. In more detail, we consider deterministic weighting schemes that at time $t$ use weight vector $\vec p^t$ and a randomly drawn candidate  $a^t \sim f(\vec \pi_{\vec \sigma^t, \vec p^t})$, where $\vec p^t$ is computed according to Algorithms~\ref{alg:full-rand} or \ref{alg:bandit-rand}. 
The key insight behind these weighting schemes is that, as we will show, if $f$ is a distribution over unilaterals, we have
\begin{equation}\label{eq:dist-unilateral}
\E_{i\sim \vec p^t} [f(\vec \pi_{\vec \sigma^t, \vec e_i})] = f(\vec \pi_{\vec \sigma^t, \vec p^t}),
\end{equation}
where the left-hand side is a vector of expectations. That is, the outcome of the voting rule $f(\vec \pi_{\vec \sigma^t, \vec p^t})$ can be alternatively implemented by applying the voting rule on the ranking of voter $i$ that is drawn at random from the distribution $\vec p^t$.
This is exactly what Algorithms~\ref{alg:full-rand} and \ref{alg:bandit-rand} do.
 Therefore, the deterministic weighting schemes induce the same distribution over alternatives at every time step as their randomized counterparts, and achieve the same regret.

\begin{theorem}
For any voting rule that is a distribution over unilaterals, there exist deterministic weighting schemes with regret of $O(\sqrt{T \ln(n)})$ and  $O(\sqrt{T n \ln(n)})$ in the full-information and partial-information settings, respectively.
\end{theorem}

\begin{proof}
Let $f$ be a distribution over unilaterals $g_1, \dots, g_k$ with corresponding probabilities $q_1, \dots, q_k$. Also, let $h_j:\L(A)\rightarrow A$ denote the function corresponding to $g_j$, for $j \in [k]$. We first prove Equation~\eqref{eq:dist-unilateral}. For ease of exposition we suppress $t$ in the notations, when it is clear from the context. Furthermore, let  $\vec \pi^i = \vec \pi_{\vec \sigma^t, \vec e_i}$. It holds that
\[ \E_{i\sim \vec p^t}\left[ f(\vec \pi_{\vec \sigma^t, \vec e_i}) \right] = \sum_{i=1}^n p^t_i f(\vec \pi^i)
=  \sum_{i=1}^n p^t_i \sum_{j=1}^k q_j \sum_{\tau\in \L(A)} \pi^i_\tau \vec e_{h_j(\tau)}  = \sum_{i=1}^n p^t_i \sum_{j=1}^k q_j  \vec e_{h_j(\sigma_i)},
\]
where the last equality follows by the fact that $\pi^i_{\sigma_i} =1$ and $\pi^i_{\tau} =0$ for any $\tau \neq \sigma_i$. 
Moreover, let $\vec \pi = \vec \pi_{\vec \sigma^t, \vec p^t}$, then
\[ f(\vec \pi_{\vec \sigma^t, \vec p^t}) = \sum_{j=1}^k  q_j  \sum_{\tau\in\L(A)} \pi_\tau \vec e_{h_j(\tau)}
= \sum_{j=1}^k  q_j  \sum_{\tau\in\L(A)} \vec e_{h_j(\tau)} \sum_{i=1}^n p^t_i \mathds{1}_{(\sigma_i = \tau)} 
= \sum_{i=1}^n p^t_i \sum_{j=1}^k  q_j  \vec e_{h_j(\sigma_i)}.
\]

Now that we have established Equation~\eqref{eq:dist-unilateral}, we use it to conclude that
\begin{align*}
\sum_{t=1}^T \U_f(\vec \pi_{\vec \sigma^t, \vec p^t}, \vec \u^t)
-  \min_{i\in[n]}  \sum_{t=1}^T \U_f(\vec \pi_{\vec \sigma^t, \vec e_{i}}, \vec \u^t)
= 
 \E\left[\sum_{t=1}^T \U_f(\vec \pi_{\vec \sigma^t, \vec e_{i}}, \vec \u^t)
-  \min_{i\in[n]}  \sum_{t=1}^T \U_f(\vec \pi_{\vec \sigma^t, \vec e_{i}}, \vec \u^t) \right], 
\end{align*}
where the expectation is taken over choice of $i\sim \vec p^t$ for all $t$. Therefore, the deterministic weighting schemes that use weight vector $\vec p^t$ achieve the same regret bounds as those established in Theorems~\ref{thm:full-rand} and \ref{thm:bandit-rand}.
\end{proof}

We have seen that there exist no-regret deterministic weighting schemes for any voting rule that is a distribution over unilaterals. It is natural to ask whether being a distribution over unilaterals is, in some sense, also a necessary condition. While we do not give a complete answer to this question, we are able to identify a sufficient condition for \emph{not} having no-regret deterministic weighting schemes. 

Recall the definitions of Condorcet winner and Condorcet consistency, introduced in Section~\ref{sec:prem-sc}. Here we extend the notion of Condorcet consistency to randomized rules. 

\begin{definition}
For a set of alternatives $A$ such that $|A|=m$, a randomized voting rule $f:\Delta(\L(A))\rightarrow \Delta(A)$ is \emph{probabilistically Condorcet consistent with gap $\delta(m)$} if for any anonymous vote profile $\vec \pi$ that has a Condorcet winner $a$, and for all alternatives $x\in A\setminus\{a\}$,
$$f(\vec \pi)_a \geq f(\vec \pi)_x + \delta(m).$$
\end{definition}
In words, a randomized voting rule is probabilistically Condorcet consistent if the Condorcet winner has strictly higher probability of being selected than any other alternative, by a gap of $\delta(m)$. As an example, the randomized Copeland rule is probabilistically Condorcet consistent with $\delta(m) = \Omega(1/m^2)$. To see why, note that for any vote profile $\vec \pi$, 
$$\sum_{a\in A}\cscore_{\vec \pi}(a)=\binom{m}{2},
$$ 
a Condorcet winner $b$ has $\cscore_{\vec \pi}(b)=m-1$, and any other alternative has score at most $m-2$. Therefore, $b$ has probability $2/m$, and any other alternative has probability at most $\frac{2(m-2)}{m(m-1)}$. Hence, we have a gap of $\frac{2}{m(m-1)}$ for the randomized Copeland rule. Also note that any deterministic voting rule that is (probabilistically) Condorcet consistent has a gap of $\delta(m) = 1$.

\begin{theorem}
\label{thm:ConLB}
For a set of alternatives $A$ such that $|A|=m$, let $f$ be a probabilistically Condorcet consistent voting rule with gap $\delta(m)$, and suppose there are $n$ voters for $n \geq 2\left(\frac{3}{2\delta(m)} + 1\right)$. Then any deterministic weighting scheme  will suffer regret of $\Omega(T)$ under $f$ (in the worst case), even in the full information setting.
\end{theorem}

We will require the following trivial lemma.

\begin{lemma} \label{sorted-sum}
Let $x_1, x_2, \cdots x_n$ be $n$ real numbers such that $x_i\geq x_{i+1}$ for all $i\in [n-1]$, and denote $S=\sum_{i=1}^n x_i$. Then for any $j\in [n]$,
$\sum_{i=1}^j x_i \geq j \frac{S}{n}$.
\end{lemma}

\begin{proof}
Assume for the sake of contradiction that there exists $j\in[n-1]$ such that $\sum_{i=1}^j x_i < j \frac{S}{n}$. It follows that there is $i\in [j]$ such that $x_i < \frac{S}{n}$. In addition, it must be the case that  $\sum_{i=j+1}^n x_i > (n-j) \frac{S}{n}$, which implies that there is $i'\in\{j+1,\ldots,n\}$ such that $x_{i'} > \frac{S}{n}$. This contradicts the fact that $x_i\geq x_{i'}$.
\end{proof}

\begin{proof}[Proof of Theorem~\ref{thm:ConLB}]
Fix an arbitrary deterministic weighting scheme. We will show that the loss of this weighting scheme is strictly higher than the average loss of the voters (for appropriately chosen vote profiles and loss functions) at every time step $t$, which directly leads to linear regret. 

Consider an arbitrary time step $t \leq T$, and let $\vec w^t$ denote the weights chosen by the weighting scheme. To construct the vote profile $\vec \sigma^t$, the adversary first partitions the voters into two sets $N_1^t$ and $N_2^t$, as follows: It sorts the weights $\vec w^t$ in non-increasing order, and then it adds voters to $N_1^t$ by their sorted weight (largest to smallest) until 
$$
W^t_1 \triangleq \sum_{i \in N_1^t} w_i^t > \frac{1}{2} \|\vec w^t\|_1,
$$ 
that is, until the voters in $N_1^t$ have more than half the total weight. The remaining voters form set $N_2^t$. 

Now, let $\tau^{x,y}\in \L(A)$ denote a ranking that places $x$ at the top (i.e., $\tau^{x,y}(x) = 1$) and $y$ in second place (i.e., $\tau^{x,y}(y) = 2$). Let $a$ and $b$ be two alternatives such that $f(\vec e_{\tau^{b,a}})_b - f(\vec e_{\tau^{b,a}})_a \geq f(\vec e_{\tau^{a,b}})_a - f(\vec e_{\tau^{a,b}})_b$, i.e., the gap between the probabilities of picking the top two alternatives in $\vec e_{\tau^{b,a}}$ is at least the corresponding gap in $\vec e_{\tau^{a,b}}$. The adversary sets the vote profile $\vec \sigma^t$ such that $\sigma_i^t = \tau^{a,b}$ for all $i \in N_1^t$ and $\sigma_i^t = \tau^{b,a}$ for all $i \in N_2^t$. Also, it sets the loss function $\vec \ell^t$ to be $\ell^t_a=1$, $\ell^t_b=0$, and $\ell^t_x=1/2$ for all $x\in A\setminus \{a,b\}$. 

Observe that for all $i\in N_1^t$, $a\succ_{\sigma_i} x$ for all $x\in A\setminus\{a\}$. Since the total weight of voters in $N_1^t$ is more than $1/2$, $a$ is a Condorcet winner in $\vec \pi_{\vec \sigma^t, \vec w^t}$. Therefore, because $f$ is probabilistically Condorcet consistent with gap $\delta(m)$, it holds that
$$f(\vec \pi_{\vec \sigma^t, \vec w^t})_a \geq f(\vec \pi_{\vec \sigma^t, \vec w^t})_b + \delta(m).$$
It follows that the loss of the weighting scheme is
\begin{equation}
\label{eq:schemeloss}
\begin{split}
L_f(\vec \pi_{\vec \sigma^t, \vec w^t}, \vec \ell^t) &= 1 \cdot f(\vec \pi_{\vec \sigma^t, \vec w^t})_a + \frac{1}{2} \cdot \left(1 - f(\vec \pi_{\vec \sigma^t, \vec w^t})_a - f(\vec \pi_{\vec \sigma^t, \vec w^t})_b\right)\\
&= \frac{1}{2} + \frac{1}{2} \left(f(\vec \pi_{\vec \sigma^t, \vec w^t})_a - f(\vec \pi_{\vec \sigma^t, \vec w^t})_b\right)\\
 &\geq \frac{1}{2} + \frac{1}{2} \delta(m).
\end{split}
\end{equation}
Similarly, the loss of voter $i$ is
\begin{equation}
\label{eq:voterloss}
L_f(\vec \pi_{\vec \sigma^t, \vec e_i}, \vec \ell^t) = L_f(\vec e_{\sigma_i^t}, \vec \ell^t) = \frac{1}{2} + \frac{1}{2} \left(f(\vec e_{\sigma_i^t})_a - f(\vec e_{\sigma_i^t})_b \right).
\end{equation}

Let $\vec q^1$ denote $f(\vec e_{\tau^{a,b}})$, i.e. the distribution over the alternatives for the votes of voters in $N_1^t$, and let $\vec q^2$ denote $f(\vec e_{\tau^{b,a}})$, i.e. the distribution over the alternatives for the votes of voters in $N_2^t$. Using these notations and Equation~\eqref{eq:voterloss}, the loss of a voter $i\in N_1^t$ is
$$L_f(\vec \pi_{\vec \sigma^t, \vec e_i}, \vec \ell^t) = \frac{1}{2} + \frac{1}{2} \left( q^1_a - q^1_b \right),$$
and the loss of a voter $i\in N_2^t$ is
$$L_f(\vec \pi_{\vec \sigma^t, \vec e_i}, \vec \ell^t) = \frac{1}{2} + \frac{1}{2} \left(q^2_a - q^2_b\right) = \frac{1}{2} - \frac{1}{2} (q^2_b - q^2_a).$$
Hence, the average loss over all voters is
\begin{align*}
L_{avg}^t &=  \frac{|N_1^t| \left(\frac{1}{2} + \frac{1}{2} \left( q^1_a - q^1_b \right)\right) + (n-|N_1^t|) \left(\frac{1}{2} - \frac{1}{2} (q^2_b - q^2_a)\right)}{n}\\
&= \frac{1}{2} + \frac{1}{2n} \left(|N_1^t|(q^1_a - q^1_b) - (n-|N_1^t|)(q^2_b - q^2_a) \right).
\end{align*}
But we chose $a$ and $b$ such that $q^1_a - q^1_b \leq q^2_b - q^2_a$. We conclude that
\begin{equation}
\label{eq:lossavg}
\begin{split}
L_{avg}^t &\leq \frac{1}{2} + \frac{1}{2n} \left(|N_1^t|(q^2_b - q^2_a) - (n-|N_1^t|)(q^2_b - q^2_a) \right)\\
&= \frac{1}{2} + \frac{1}{2} (q^2_b - q^2_a) \frac{(2|N_1^t| - n)}{n}.
\end{split}
\end{equation}

Our goal is to derive an upper bound on the expression $\frac{1}{2} (q^2_b - q^2_a) \frac{(2|N_1^t| - n)}{n}$. Specifically, we wish to prove that 
\begin{equation}
\label{eq:cases}
\frac{1}{2} (q^2_b - q^2_a) \frac{(2|N_1^t| - n)}{n}\leq \frac{\delta(m)}{3}.
\end{equation}
We do this by examining two cases. 

\paragraph{Case 1: $W^t_1 \geq \left(\frac{1}{2} + \frac{\delta(m)}{3}\right) \|\vec w^t\|_1$.}
Informally, this is the case when the weights of $N_1^t$  overshot $\|\vec w^t\|_1/2$ by a fraction of at least $\delta(m)/3$. This means that the last voter added to $N_1^t$ has a weight of at least $W^t_1 - \frac{\|\vec w^t\|_1}{2}$. Since the weights were added in non-increasing order, it follows that each voter in $N_1^t$ has a weight of at least $W^t_1 - \frac{\|\vec w^t\|_1}{2}$. Therefore, 
$$W^t_1 = \sum_{i \in N_1^t} w_i^t \geq \sum_{i \in N_1^t} \left(W^t_1 - \frac{\|\vec w^t\|_1}{2}\right) = |N_1^t| \left(W^t_1 - \frac{\|\vec w^t\|_1}{2}\right),$$
or equivalently,
\begin{equation}
\label{eq:size}
|N_1^t| \leq \frac{1}{1-\frac{\|\vec w^t\|_1}{2 W^t_1}}.
\end{equation}
We have also assumed that $\frac{W^t_1}{\|\vec w^t\|_1} \geq \left(\frac{1}{2} + \frac{\delta(m)}{3}\right)$. Using Equation~\eqref{eq:size}, we obtain
\begin{equation}
\label{eq:n1}
|N_1^t| \leq\frac{1}{1-\frac{1}{1 + \frac{2\delta(m)}{3}}}= \frac{3}{2\delta(m)} + 1.
\end{equation}

Let us now examine the expression on the left-hand side of Equation~\eqref{eq:cases}. 
Note that $b$ is a Condorcet winner in $\vec e_{\tau^{b,a}}$. Hence, $q^2_b \geq q^2_a + \delta(m)$, and, in particular, $q^2_b - q^2_a > 0$. In addition, we have assumed that $n \geq 2(\frac{3}{2\delta(m)} + 1)$, which implies (by Equation~\eqref{eq:n1}) that $n \geq 2|N_1^t|$. It follows that 
$$
\frac{1}{2} (q^2_b - q^2_a) \frac{(2|N_1^t| - n)}{n}\leq 0\leq \frac{\delta(m)}{3},
$$
thereby establishing Equation~\eqref{eq:cases} for this case. 

\paragraph{Case 2: $W^t_1 < \left(\frac{1}{2} + \frac{\delta(m)}{3}\right) \|\vec w^t\|_1$}
Since $N_1^t$ contains voters who have the largest $|N_1^t|$ weights, Lemma~\ref{sorted-sum} implies that
$$W^t_1 = \sum_{i \in N_1^t} w_i^t \geq |N_1^t| \frac{\|\vec w^t\|_1}{n}.$$
We have also assumed that $W^t_1 < (\frac{1}{2} + \frac{\delta(m)}{3}) \|\vec w^t\|_1$. Combining the last two inequalities, we obtain 
\begin{equation}
\label{eq:n1again}
|N_1^t| < n \left(\frac{1}{2} + \frac{\delta(m)}{3}\right).
\end{equation}

Let us examine, once again, the left-hand side of Equation~\eqref{eq:cases}. Recall that $q^2_b - q^2_a > 0$, because $b$ is a Condorcet winner in $\tau^{b,a}$. So, if $2|N_1^t| - n \leq 0$, then Equation~\eqref{eq:cases} clearly holds, as in Case 1. And if $2|N_1^t| - n > 0$, the equation also holds, because
\begin{align*}
\frac{1}{2} (q^2_b - q^2_a) \frac{(2|N_1^t| - n)}{n}
\leq \frac{1}{2} \cdot 1 \cdot \frac{(2|N_1^t| - n)}{n} = \frac{|N_1^t|}{n}-\frac{1}{2}< \frac{\delta(m)}{3},
\end{align*}
where the last inequality follows from Equation~\eqref{eq:n1again}.

\medskip

To complete the proof, we combine Equations \eqref{eq:schemeloss}, \eqref{eq:lossavg}, and \eqref{eq:cases}, to obtain
$$
L_f(\vec \pi_{\vec \sigma^t, \vec w^t}, \vec \ell^t) \geq L_{avg}^t + \frac{\delta(m)}{6}.
$$
The best voter in hindsight incurs loss that is at most as high as the average voter. Therefore, the overall regret is
\begin{align*}
Reg_T &= \sum_{t=1}^T L_f(\vec \pi_{\vec \sigma^t, \vec w^t}, \vec \ell^t) - \min_i \sum_{t=1}^T L_f(\vec \pi_{\vec \sigma^t, \vec e_i}, \vec \ell^t)\\
&\geq \sum_{t=1}^T L_f(\vec \pi_{\vec \sigma^t, \vec w^t}, \vec \ell^t) - \sum_{t=1}^T L_{avg}^t\\
&\geq T \frac{\delta(m)}{6}.
\end{align*}
In words, the weighting scheme suffers linear regret.
\end{proof}

It is interesting to note that Theorems~\ref{thm:det-voting-wt} and \ref{thm:ConLB} together imply that distributions over unilaterals are not probabilistically Condorcet consistent. This is actually quite intuitive: Distributions over unilaterals are ``local'' in that they look at each voter separately, whereas Condorcet consistency is a global property. In fact, these theorems can be used to prove --- in an especially convoluted and indirect way --- a simple result from social choice theory~\cite{Moul83}: No positional scoring rule is Condorcet consistent!  
\section{Discussion}

We conclude by discussing several conceptual points.

\paragraph{Changing the sets of alternatives and voters over time} We wish to emphasize that the set of alternatives at each time step, i.e., in each election, can be completely different. Moreover, the \emph{number} of alternatives could be different. In fact, our positive results do not even depend on the number of alternatives $m$, so we can simply set $m$ to be an upper bound. By contrast, we do need the set of voters to stay fixed throughout the process, but this is consistent with our motivating examples (e.g., a group of partners in a small venture capital firm would face different choices at every time step, but the composition of the group rarely changes).

\paragraph{Optimizing the voting rule} Throughout the paper, the voting rule is exogenous. One might ask whether it makes sense to optimize the choice of voting rule itself, in order to obtain good no-regret learning results. Our answer is ``yes and no''. On the one hand, we believe our results do give some guidance on choosing between voting rules. For example, from this viewpoint, one might prefer randomized Borda (which admits no-regret algorithms under a deterministic weighting scheme) to randomized Copeland (which does not). On the other hand, many considerations are factored into the choice of voting rule: social choice axioms, optimization of additional objectives~\cite{PSZ16,BCHL+15,EFS09,CS05b}, and simplicity. It is therefore best to think of our approach as \emph{augmenting} voting rules that are already in place. 

\paragraph{A natural, harder benchmark} In our model (see Section~\ref{sec:problem}), we are competing with the best voter in hindsight. But our action space consists of \emph{weight vectors}. It is therefore natural to ask whether we can compete with the best weight vector in hindsight. Clearly this alternative benchmark is at least as hard, because the best voter $i^*$ corresponds to the weight vector $\mathbf{e}_{i^*}$. Informally, the alternative benchmark is strictly harder if the voting rule does not nicely decompose across voters (like distributions over unilaterals do). We can prove some positive results for the alternative benchmark under specific voting rules (such as randomized Copeland) and specific families of weight vectors; but properly dealing with it largely remains an open problem.

\subsection*{Acknowledgments}
The authors were partially supported by the National Science Foundation under grants IIS-1350598 and CCF-1525932, by the Office of Naval Research, and by a Sloan Research Fellowship. Haghtalab was partially supported by a Microsoft Research PhD Fellowship.

\bibliographystyle{plainnat}
\bibliography{abb,ultimate}

\end{document}